\documentclass[journal]{IEEEtran}

\ifCLASSINFOpdf
\else
\fi

\usepackage{url}

\usepackage{graphics} 
\usepackage{mathtools} 
\usepackage{amsthm} 
\usepackage{amssymb}  
\usepackage{color}
\usepackage{algpseudocode}
\usepackage{algorithm}
\usepackage{graphicx}
\usepackage{verbatim}
\graphicspath{{Figures/}}
\newtheorem{thm}{Theorem}[section]

\newcommand{\Ad}{\tilde{A}}

\DeclareMathOperator*{\Es}{\mathrm{E}}

\DeclareMathOperator*{\argmin}{{ \sf{argmin}}}

\newcommand{\s}{x}
\newcommand{\ug}{u}
\DeclareMathOperator{\y}{\mathrm{y}}
\DeclareMathOperator{\Traj}{\mathbf{x}}

\DeclareMathOperator{\ntimes}{\otimes}

\newcommand{\sv}[1]{\sigma\left({#1}\right)}
\newcommand{\Reg}[1]{R\left({#1}\right)}
\newcommand{\ord}[2]{{#1}_{\left[{#2}\right]}}
\newcommand{\abs}[1]{\left|  {#1}  \right|}

\newcommand{\nc}{n_u}
\DeclareMathOperator*{\mini}{\text{Minimize }}

\newcommand{\Kb}{\mathbf{K}}

\newcommand{\tran}[1]{{#1}^T}

\newcommand{\ns}{n}

\newcommand{\powb}[2]{{\left({#1}\right)}^{{#2}}}

\newcommand{\ExP}[2]{\Es_{{#1}}\left [{#2}\right ]}

\newcommand{\expb}[1]{\exp\left({#1}\right)}
\newcommand{\logb}[1]{\log\left({#1}\right)}

\newcommand{\norm}[1]{\lVert {#1} \rVert}
\newcommand{\tr}[1]{\mathrm{tr}\left( {#1} \right)}

\newcommand{\N}{\mathcal{N}}
\newcommand{\R}{\mathbf{R}}

\newcommand{\inv}[1]{{\left(#1\right)}^{-1}}
\newcommand{\inve}[1]{{#1}^{-1}}

\newcommand{\Gauss}[2]{\N\left({#1},{#2}\right)}

\newcommand{\lai}[2]{\lambda_{#2}\left({#1}\right)}
\newcommand{\lam}[1]{\lambda_{\min}\left({#1}\right)}
\newcommand{\lama}[1]{\lambda_{\max}\left({#1}\right)}

\newcommand{\svai}[2]{\sigma_{#2}\left({#1}\right)}
\newcommand{\svam}[1]{\sigma_{\min}\left({#1}\right)}
\newcommand{\svama}[1]{\sigma_{\max}\left({#1}\right)}

\newcommand{\opt}[1]{{#1}^\ast}

\newcommand{\Htwo}{\mathcal{H}_2}
\newcommand{\Hinf}{\mathcal{H}_\infty}
\newcommand{\LQR}{\mathrm{LQR}}
\newcommand{\C}{\mathcal{C}}

\newcommand{\traj}{\mathbf{x}}
\newcommand{\detb}[1]{\det\left({#1}\right)}
\newcommand{\tranb}[1]{{\left({#1}\right)}^T}
\newcommand{\br}[1]{\left({#1}\right)}

\newcommand{\noise}{w}
\newcommand{\ntraj}{\mathbf{w}}
\newcommand{\Acal}[2]{F\left({#1}\right)}
\newcommand{\Aca}{F}

\newcommand{\wtraj}{\mathbf{w}}

\newcommand{\Dyna}{a}
\newcommand{\DynB}{B}
\newcommand{\Dynf}{\phi}

\newcommand{\nm}{\mathrm{m}}
\newcommand{\nh}{\mathrm{k}}

\newcommand{\FHinf}[1]{q_\infty}
\newcommand{\FHtwo}[1]{q_2}
\newcommand{\FHone}[1]{q_1}

\newcommand{\FCtwo}[1]{q^c_2\left({#1}\right)}
\newcommand{\FCinf}[1]{q^c_\infty\left({#1}\right)}

\newcommand{\CON}{\rm{CON}}
\newcommand{\Mih}{\rm{NCON}}

\newcommand{\OPT}{\rm{OPT}}

\newcommand{\optt}[1]{{{#1}_c}^\ast}

\newcommand{\Vari}[1]{\mathrm{Var}\left({#1}\right)}

\begin{document}
%
\title{\LARGE \bf
Convex Structured Controller Design
}
%
%
%

\author{Krishnamurthy~Dvijotham,
        Emanuel~Todorov,
        and~Maryam~Fazel
\thanks{K. Dvijotham is with the Department
of Computer Science and Engineering, University of Washington, Seattle,
WA, 98195 USA e-mail: dvij@cs.washington.edu}
\thanks{E. Todorov is with the Departments of Computer Science and Engineering and Applied Mathematics at the University of Washington, Seattle. e-mail: todorov@cs.washington.edu}
\thanks{M. Fazel is with the Department of Electrical Engineering at the University of Washington, Seattle. e-mail: mfazel@uw.edu}%
}

\markboth{Submitted to IEEE Transactions on Networked Control Systems}%
{Shell \MakeLowercase{\textit{et al.}}: Bare Demo of IEEEtran.cls for Journals}
%



\maketitle

\IEEEpeerreviewmaketitle
\begin{abstract}
We consider the problem of synthesizing optimal linear feedback policies subject to arbitrary convex constraints on the feedback matrix. This is known to be a hard problem in the usual formulations ($\Htwo,\Hinf,\LQR$) and previous works have focused on characterizing classes of structural constraints that allow efficient solution through convex optimization or dynamic programming techniques. In this paper, we propose a new control objective and show that this formulation makes the problem of computing optimal linear feedback matrices convex under arbitrary convex constraints on the feedback matrix. This allows us to solve problems in decentralized control (sparsity in the feedback matrices), control with delays and variable impedance control. Although the control objective is nonstandard, we present theoretical and empirical evidence that it agrees well with standard notions of control.  We also present an extension to nonlinear control affine systems. We present numerical experiments validating our approach.
\end{abstract}


\section{INTRODUCTION}
Linear feedback control synthesis is a classical topic in control theory and has been extensively studied in the literature. From the perspective of stochastic optimal control theory, the classical result is the existence of an optimal linear feedback controller for systems with linear dynamics, quadratic costs and gaussian noise (LQG systems) that can be computed via dynamic programming  \cite{kalman1960new}. However, if one imposes additional constraints on the feedback matrix (such as a sparse structure arising from the need to implement control in a decentralized fashion), the dynamic programming approach is no longer applicable. In fact, it has been shown that the optimal control policy may not even be linear \cite{witsenhausen1968counterexample} and that the general problem of designing linear feedback gains subject to constraints is NP-hard \cite{blondel1997np}.\\
Previous approaches to synthesizing structured controllers can be broadly categorized into three types: Frequency Domain Approaches\cite{rotkowitz2002decentralized}\cite{qi2004structured}\cite{rotkowitz2006}\cite{shah2013}, Dynamic Programming Approaches \cite{fan1994centralized}\cite{swigart2010explicit}\cite{lamperski2013optimal} and Nonconvex optimization methods \cite{linfarjovTAC13admm}\cite{apkarian2008mixed}\cite{burke2006hifoo}. The first two classes of approaches find \emph{exact} solutions to structured control problems for special cases. The third class of approaches tries to directly solve the optimal control problem (minimizing the $\Htwo$,$\Hinf$ norm) subject to constraints on the controller, using nonconvex optimization techniques. These are generally applicable, but are susceptible to local minima and slow convergence (especially for nonsmooth norms such as $\Hinf$).\\
In this paper, we take a different approach: We reformulate the structured control problem using a family of new control objectives (section \ref{sec:Problem}). We develop these bounds as follows: The $\Htwo$ and $\Hinf$ norms can be expressed as functions of singular values of the linear mapping from disturbance trajectories to state trajectories. This mapping is a highly nonlinear function of the feedback gains. However, the inverse of this mapping has a simple linear dependence on the feedback gains. Further, the determinant of the mapping has a fixed value independent of the closed loop dynamics - this is in fact a finite horizon version of Bode's sensitivity integral and has been studied in \cite{iglesias2001tradeoffs}. By exploiting both these facts, we develop upper bounds on the $\Htwo,\Hinf$ norms in terms of the singular values of the inverse mapping.  We show that these upper bounds have several properties that make them desirable control objectives. For the new family of objectives, we show that the resulting problem of designing an optimal linear state feedback matrix, under arbitrary convex constraints, is convex (section \ref{sec:Main}). Further, we prove suboptimality bounds on how the solutions of the convex problems compare to the optima of the original problem. Our approach is directly formulated in state space terminology and does not make any reference to frequency domain concepts. Thus, it applies directly to time-varying systems. We validate our approach numerically and show that the controllers synthesized by our approach achieve good performance (section \ref{sec:NumExpt}).\\
The work presented here is an extension of a recent conference publication \cite{Dj2013}. However, this paper contains a significant reformulation of the results presented there and also has new results. The conference paper \cite{Dj2013} was formulated in terms of the eigenvalues of the covariance matrix of trajectories. In this paper, we look at the linear map from noise to state trajectories (denoted by $\Aca$ in this paper), which is a Cholesky factor of the covariance matrix. This allows us to produce simpler proofs of convexity and deal with a more general class of objectives. For example, the nuclear norm (and more generally Ky-Fan norms) of $\Aca$ is a valid objective in our formulation while it was not in \cite{Dj2013}, since it is equal to the sum of square roots of eigenvalues of the covariance matrix which is a nonconvex function. Further, we provide an analysis quantifying how well the solutions to our convex objectives perform in terms of the original nonconvex $\Htwo,\Hinf$ objectives. We present numerical results comparing results of our formulation to other nonconvex approaches for structured controller synthesis. Finally, we also present a generalization of our approach to nonlinear systems.

\section{PROBLEM FORMULATION} \label{sec:Problem}
Consider a finite-horizon discrete-time linear system in state-space form:
\begin{align*}
& \s_1 = D_0\noise_0 \\
& \s_{t+1} = A_t\s_t + B_t \ug_t + D_{t}\noise_{t}, \quad t=1,2,\ldots,N-1.
\end{align*}
Here $t=0,1,2,\ldots,N$ is the discrete time index, $\s_t \in \R^n$ is the plant state, $\noise_t \in \R^n$ is an exogenous disturbance and $\ug_t \in \R^{\nc}$ is the control input. We employ static state feedback:
\[\ug_t=K_t\s_t.\]
Let $\Kb=\{K_t: t=1,2,\ldots,N-1\}$ and denote the closed-loop system dynamics by
\[\Ad_t(K_t)=A_t+B_tK_t.\]
Let $\lama{M}$ denote the maximum eigenvalue of an $l \times l$ symmetric matrix $M$, $\lam{M}$ the minimum eigenvalue and $\lai{M}{i}$ the $i$-th eigenvalue in descending order:
 \begin{align*}
&  \lai{M}{l}=\lam{M}\leq \lai{M}{l-1} \leq  \ldots \leq \lama{M}=\lai{M}{1}.
 \end{align*}
Similarly, singular values of a general rank $l$ matrix $M$ are:
\begin{align*}
& \svai{M}{l} =\svam{M}\leq \svai{M}{2} \leq \ldots \leq \svama{M}=\svai{M}{1}.
\end{align*}
$I$ denotes the identity matrix. Boldface smallcase letters denote trajectories:
 \[\traj=\begin{pmatrix} \s_1 \\ \vdots \\ \s_N \end{pmatrix},\wtraj=\begin{pmatrix} \noise_0 \\ \vdots \\ \noise_{N-1} \end{pmatrix}\]
 For $z \in \R^n$,
 \[\Vari{z}=\frac{1}{n}\sum_{i=1}^n \powb{z_i-\frac{\sum_{i=1}^n z_i}{n}}{2}.\]
$\ord{z}{i}$ is the $i$-th largest component of $z$ and $|z|$ the vector with entries $|z_1|,\ldots,|z_n|$. Finally, $\N(\mu,\Sigma)$ denotes a Gaussian distribution with mean $\mu$ and covariance matrix $\Sigma$.

There is a linear mapping between disturbance and state trajectories for a linear system. This will play a key role in our paper, and we denote it by
\begin{align*}
& \traj = \Acal{\Kb}{N} \ntraj,\\
& \text{where } \Acal{\Kb}{N}= \\
& \left[\begin{matrix}
  D_0               & 0         & \ldots        & \quad 0 \\
  \Ad_1 D_0         & D_1       & \ldots        & \quad 0 \\
  \Ad_2\Ad_1 D_0    & \Ad_2 D_1 & \ldots        & \quad 0 \\
  \vdots            & \vdots    & \vdots        & \quad \vdots\\
  \prod_{\tau=1}^{N-1} \Ad_{N-\tau}D_0 & \prod_{\tau=2}^{N-1} \Ad_{N-\tau}D_1 & \ldots & \quad D_{N-1}
 \end{matrix}\right].
 \end{align*}
Our formulation differs from standard control formulations in the following ways:
 \begin{itemize}
\item[1] We assume that the controller performance is measured in terms the norm of the system trajectory $\tran{\traj}\traj$ (see section \ref{sec:PenControl} for an extension that includes control costs).
\item[2] As mentioned earlier, we restrict ourselves to have static state feedback $\ug_t=K_t\s_t$ (section \ref{sec:DynOutput} discusses dynamic output feedback).
\item[3] We assume that $D_t$ is square and invertible.
\end{itemize}
Finite-horizon versions of the $\Htwo$ and $\Hinf$ norms of the system are given by:
\begin{align}
& \FHtwo{N}(\Kb) = \ExP{\noise_t \sim \Gauss{0}{I}}{\sum_{t=1}^N \tran{\s_t}\s_t}=\ExP{}{\tr{\traj\tran{\traj}}} \nonumber \\
& =\tr{\tran{\Acal{\Kb}{N}}\Acal{\Kb}{N}} = \sum_{i=1}^{nN} {\svai{\Acal{\Kb}{N}}{i}}^2 \label{eq:Ftwo} \\
& \FHinf{N}(\Kb) = \sqrt{\max_{\wtraj \neq 0} \frac{\sum_{t=1}^N \tran{\s_t}{\s_t}}{\sum_{t=0}^{N-1} \tran{\noise_t}{\noise_t}}} \nonumber \\
& =\max_{\wtraj \neq 0} \frac{\norm{\Acal{\Kb}{N}\ntraj}}{\norm{\ntraj}}=\svama{\Acal{\Kb}{N}}. \label{eq:Finf}
\end{align}
If there are no constraints on $\Kb$, these problems can be solved using standard dynamic programming techniques. However, we are interested in synthesizing structured controllers. We formulate this very generally: We allow \emph{arbitrary} convex constraints on the set of feedback matrices: $\Kb \in \C$ for some convex set $\C$. Then, the control synthesis problem becomes
\begin{align}
& \mini_{\Kb \in \C} \FHtwo{N}(\Kb)\label{eq:TwoOrig} \\
& \mini_{\Kb \in \C} \FHinf{N}(\Kb)\label{eq:InfOrig}
\end{align}
The general problem of synthesizing stabilizing linear feedback control, subject even to simple bound constraints on the entries of $K$, is known to be hard \cite{blondel1997np}. Several hardness results on linear controller design can be found in \cite{blondel2000survey}. Although these results do not cover the problems \eqref{eq:TwoOrig}\eqref{eq:InfOrig}, they suggest that \eqref{eq:TwoOrig}\eqref{eq:InfOrig} are hard optimization problems. In this paper, we propose an alternate objective function based on the singular values of the inverse mapping $\inve{\Acal{\Kb}{N}}$ and prove that this objective can be optimized using convex programming techniques under \emph{arbitrary} convex constraints on the feedback matrices $\Kb=\{K_t\}$. Given the above hardness results, it is clear that the optimal solution to the convex problem will not match the optimal solution to the original problem. However, we present theoretical and numerical evidence to suggest that the solutions of the convex problem we propose (\eqref{eq:ObjHtwo},\eqref{eq:ObjHinf}) approximate the solution to the original problems (\eqref{eq:TwoOrig},\eqref{eq:InfOrig}) well for several problems.

\subsection{Control Objective} \label{sec:Just}

The problems \eqref{eq:TwoOrig},\eqref{eq:InfOrig} are non-convex optimization problems, because of the nonlinear dependence of $\Acal{\Kb}{N}$ on $\Kb$. In this section, we will derive convex upper bounds on the singular values of $\Acal{\Kb}{N}$ that can be optimized under arbitrary convex constraints $\C$. We have the following results (section \ref{sec:Appendix}, theorems \ref{thm:HinfUB}, \ref{thm:HtwoUB}):
\begin{align*}
& \FHinf{N}(\Kb) \leq \br{\prod_{t=0}^{N-1} \detb{D_t}}\powb{\frac{\sum_{i=1}^{\ns N-1} \svai{\inve{\Acal{\Kb}{N}}}{i}}{\ns N-1}}{\ns N-1} \\
& \FHtwo{N}(\Kb) \leq nN\br{\prod_{t=0}^{N-1} \detb{D_t}}^{2}\powb{\svama{\inve{\Acal{\Kb}{N}}}}{2(\ns N-1)}
\end{align*}
\begin{figure}\begin{center}
\includegraphics[width=.7\columnwidth]{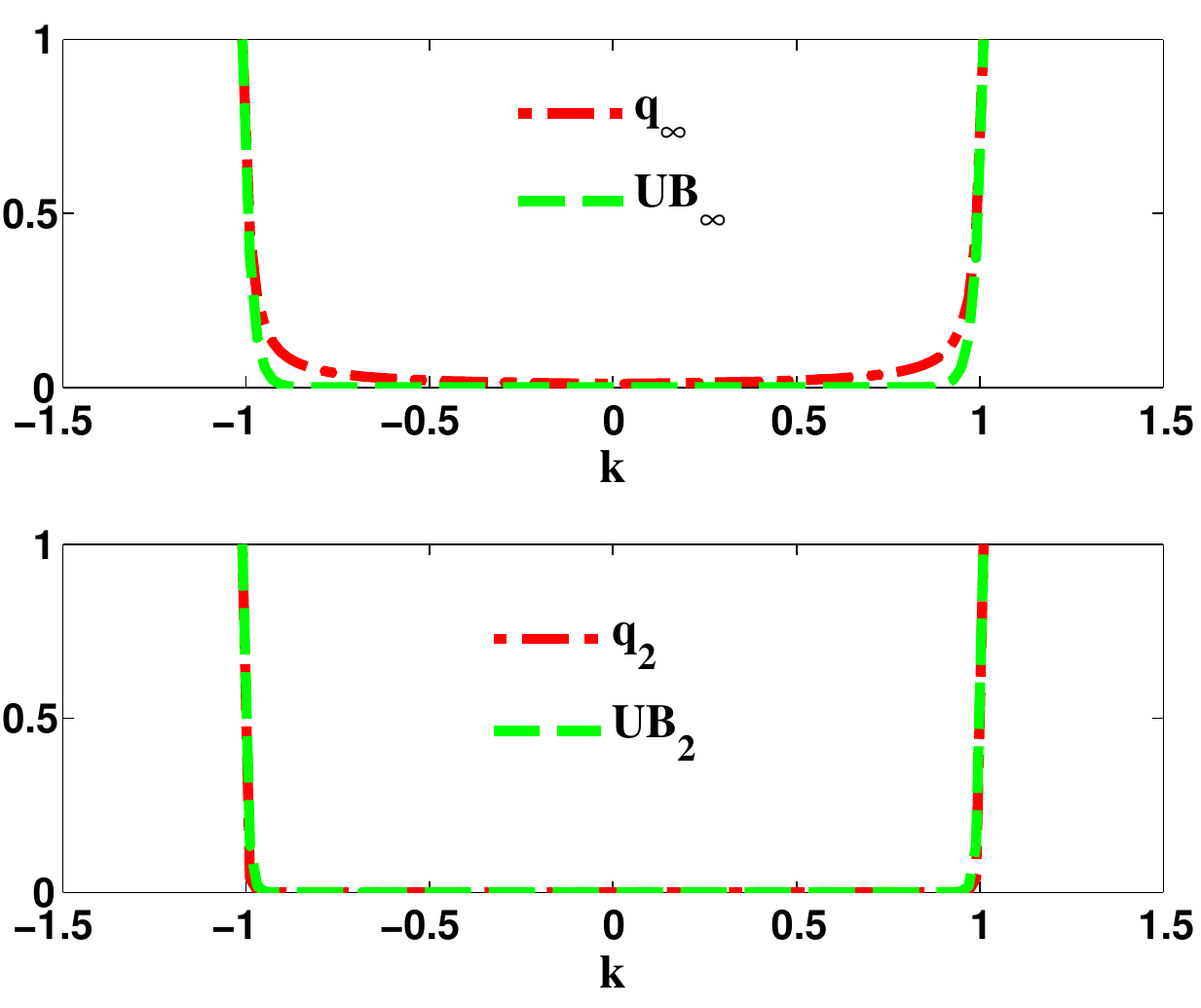}
\caption{Convex Surrogate vs Original Objective (rescaled to lie in [0,1]):  $\FHinf{}(k)$ vs $UB_\infty(k)$ (top),$\FHtwo{}(k)$ vs $UB_2(k)$ (bottom)} \label{fig:ObjPlot}
\end{center}\end{figure}
To illustrate the behavior of these upper bounds (denoted $UB_\infty,UB_2$), we plot them for a scalar linear system
\[\s_{t+1}=\ug_t+\noise_t,\s_t,\ug_t=k\s_t,k \in \R\]
over a horizon $N=100$ in figure \ref{fig:ObjPlot}. When $|k|<1$, this system is unstable and otherwise it is stable. Thus, $|k|$ is a measure of the ``degree of instability'' of the system. As expected, the original objectives grow slowly to the point of instability and then blow up. The convex upper bounds are fairly loose upper bounds and increase steadily. However, the rate of growth increases with degree of instability. Similar results are observed for $n>1$.

\subsection{A General Class of Control Objectives}
Inspired by the upper bounds of the previous section, we formulate the controller design problem as follows:
\begin{align}
& \mini_{\Kb \in \C}               \FCtwo{\Kb}=\svama{\inve{\Acal{\Kb}{N}}} \text{ (surrogate to $\FHtwo{N})$}\label{eq:ObjHtwo} \\
& \mini_{\Kb \in \C}               \FCinf{\Kb}=\sum_{i=1}^{\ns N-1} \svai{\inve{\Acal{\Kb}{N}}}{i} \text{ (surrogate to $\FHinf{N})$}\label{eq:ObjHinf}
\end{align}
The objectives \eqref{eq:ObjHtwo},\eqref{eq:ObjHinf} are just two of the control objectives that are allowed in our framework. We can actually allow a general class of objectives that can be minimized for control design. From \cite{lewis1995convex}, we know that for any \emph{absolutely invariant} convex function $f(x)$ on $\R^n$, the function $g(X)=f(\sv{X})$ on $\R^{n \times n}$ is convex. This motivates us to consider a generalized control objective:
\begin{align}
& \mini_{\Kb}  \underbrace{f\left(\sv{\inve{\Acal{\Kb}{N}}}\right)}_{\text{Controller Performance}}+\underbrace{\Reg{\Kb}}_{\text{Minimize Control Effort}}\nonumber \\
& \text{Subject to } \Kb \in \C \label{eq:ObjGen}
\end{align}
where $\C$ is a convex set encoding the structural constraints on $\Kb$ and $\Reg{\Kb}$ is a convex penalty on the feedback gains $\Kb$. We show (in theorem \ref{thm:ConvObj}) that this problem is a convex optimization problem.
Common special cases for $f$ are:
\begin{itemize}
\item[1] $f(x) = \norm{x}_\infty$ which gives rise to the spectral norm $\norm{\inv{\Acal{\Kb}{N}}}=\svama{\inv{\Acal{\Kb}{N}}}$, the same as \eqref{eq:ObjHtwo}.
\item[2] $f(x) = \norm{x}_1$ which gives rise to the nuclear norm $\norm{\inv{\Acal{\Kb}{N}}}_*=\sum_{i} \svai{\inv{\Acal{\Kb}{N}}}{i}$.
\item[3] $f(x) = \sum_{i=1}^{k} \ord{\abs{x}}{i}$ which gives rise to the Ky Fan k-norm $\sum_{i=1}^{k} \svai{\inv{\Acal{\Kb}{N}}}{i}$. In particular $f(x) = \sum_{i=1}^{nN-1} \ord{\abs{x}}{i}$ corresponds to \eqref{eq:ObjHinf}.
\end{itemize}
 A common choice for $\Reg{\Kb}$ is $\norm{\Kb}^2$. For decentralized control, $\C$ would be of the form $\C=\{\Kb: {\Kb_t} \in S\}$
where $S$ is the set of matrices with a certain sparsity pattern corresponding to the decentralization structure required. We now present our main theorem proving the convexity of the generalized problem \eqref{eq:ObjGen}.

\section{MAIN TECHNICAL RESULTS} \label{sec:Main}

\subsection{PROOF OF CONVEXITY}
\begin{thm} \label{thm:ConvObj}
If $f$ is an absolutely symmetric lower-semicontinuous convex function, $\Reg{\Kb}$ is a convex function and $\C$ is a convex set, then the problem \eqref{eq:ObjGen} is a convex optimization problem.
\end{thm}

\begin{proof}
The proof relies on the structure of $\inve{\Acal{\Kb}{N}}$. Rewriting the discrete-time dynamics equations, we have: \[\noise_0=\inve{D_0}\s_1,\noise_t=\inve{D_t}\s_{t+1}-\inve{D_t}\tilde{A}_t\s_t \text{ for } t\geq 1.\]
It can be shown that $\inve{\Acal{\Kb}{N}}$ is given by
\[\left[\begin{matrix}
  \inve{D_0}             & 0         & \ldots    & \ldots & \quad 0 \\
  -\inve{D_1}\Ad_1        & \inve{D_1}         & \ldots    & \ldots & \quad 0 \\
  0             & -\inve{D_2}\Ad_2    & \inve{D_2}         & \ldots & \quad 0 \\
  \vdots        & \vdots    & \vdots    & \cdots & \quad \vdots\\
  0             & 0         & 0         & \ldots & \quad \inve{D_{N-1}}
 \end{matrix}\right] \]
This can be verified by simple matrix multiplication. Now, the convexity is obvious since $\Ad_t=A_t+B_tK_t$ is a linear function of $\Kb$, and so is $\inve{\Acal{\Kb}{N}}$. Since $f$ is an absolutely symmetric lower-semicontinuous convex function, $f(\sv{X})$ is a convex function of $X$ \cite{lewis1995convex}. Thus, $f(\sv{\inv{\Acal{\Kb}{N}}})$ is the composition of an affine function in $\Kb$ with a convex function, and is hence convex. The function $\Reg{\Kb}$ is known to be convex and so are the constraints $\Kb \in \C$. Hence, the overall problem is a convex optimization problem. \end{proof}

\subsection{SUBOPTIMALITY BOUNDS}
We are using convex surrogates for the $\FHtwo{N},\FHinf{N}$ norms. Thus, it makes sense to ask the question: How far are the optimal solutions to the convex surrogates from those of the original problem? We answer this question by proving multiplicative suboptimality bounds: We prove that the ratio of the $\FHtwo{N}$ norm of the convex surrogate solution and the $\FHtwo{N}$-optimal solution is bounded above by a quantity that decreases as the variance of the singular vector of $\inv{\Acal{\Kb}{N}}$ at the optimum. Although these bounds may be quite loose, they provide qualitative guidance about when the algorithm would perform well.
\begin{thm}\label{thm:HtwoSUB}
Let the solution to the convex optimization and original problem be:
\[\optt{\Kb}=\argmin_{\Kb \in \C} \svama{\inv{\Acal{\Kb}{N}}} \textrm{ (Convex Opt)} \]
\[\opt{\Kb}=\argmin_{\Kb \in \C} \sum_{i} \powb{\svai{\Acal{\Kb}{N}}{i}}{2} \textrm{ (Original Opt)}\]
respectively. Let $\opt{\Aca}=\inv{\Acal{\opt{\Kb}}{N}}, \optt{\Aca}=\inv{\Acal{\optt{\Kb}}{N}}$. Let
\[\optt{\sigma}=\left[\powb{\frac{\svai{\optt{\Aca}}{2}}{\svai{\optt{\Aca}}{nN}}}{2},\ldots,\powb{\frac{\svai{\optt{\Aca}}{2}}{\svai{\optt{\Aca}}{2}}}{2}\right]\] \[\opt{\sigma}=\left[\powb{\frac{\svai{\opt{\Aca}}{nN}}{\svai{\opt{\Aca}}{nN}}}{2},\ldots,\powb{\frac{\svai{\opt{\Aca}}{nN}}{\svai{\opt{\Aca}}{2}}}{2}\right]\]
Then,
\[\frac{\FHtwo{N}(\optt{\Kb})}{\FHtwo{N}(\opt{\Kb})} \leq \br{\frac{nN}{nN-1}}\expb{\frac{\Vari{\optt{\sigma}}-\Vari{\opt{\sigma}}}{2}}\]
\end{thm}
\begin{proof}
The proof relies on Holder's defect formula which quantifies the gap in the AM-GM inequality \cite{becker2012variance}. For any numbers $0<a_m\leq \ldots \leq a_1$, we have:
\begin{align*}
&\br{\frac{\sum_{i=1}^{m} a_i}{m}}\expb{-\frac{\mu}{2}\Vari{a}} = \powb{\prod_{i=1}^m a_i }{1/m} \\
\end{align*}
where $\mu \in \left[\powb{\frac{1}{a_1}}{2},\powb{\frac{1}{a_m}}{2}\right]$.  Plugging in the lower and upper bounds for $\mu$, we get
\begin{align*}
&\br{\frac{\sum_{i=1}^{m} a_i}{m}}\expb{-\frac{\Vari{a/a_1}}{2}} \geq \powb{\prod_{i=1}^m a_i }{1/m} \\
&\br{\frac{\sum_{i=1}^{m} a_i}{m}}\expb{-\frac{\Vari{a/a_m}}{2}} \leq \powb{\prod_{i=1}^m a_i }{1/m}.
\end{align*}
Using this inequality with $a_i=\powb{\svai{\opt{\Aca}}{nN-i+1}}{-2},i=1,2,3,\ldots,nN-1$, we get
\begin{align*}
& \frac{\FHtwo{N}(\opt{\Kb})}{nN-1} \geq \frac{1}{nN-1}\sum_{i=2}^{nN} \frac{1}{\powb{{\svai{\opt{\Aca}}{i}}}{2}} \\
& \geq \expb{\frac{\Vari{\opt{\sigma}}}{2}}\powb{\prod_{i=2}^{nN} \frac{1}{\powb{\svai{\opt{\Aca}}{i}}{2}}}{\frac{1}{nN-1}} \\
& = c\expb{\frac{\Vari{\opt{\sigma}}}{2}}\powb{\svama{\opt{\Aca}}}{\frac{2}{nN-1}}
\end{align*}
where $c=\powb{\prod_{t=0}^{N-1}\detb{D_t}}{\frac{2}{nN-1}}$ and the last equality follows since $\detb{\opt{\Aca}}=\prod_{t=0}^{N-1}\detb{D_t}$. Since $\optt{\Kb}$ minimizes $\svama{\inve{\Acal{\Kb}{N}}}$, we have
\begin{align*}
& \frac{\FHtwo{N}(\opt{\Kb})}{nN-1} \geq c\expb{\frac{\Vari{\opt{\sigma}}}{2}}\powb{\svama{\optt{\Aca}}}{\frac{2}{nN-1}} \\
& \geq \expb{\frac{\Vari{\opt{\sigma}}}{2}}\powb{\prod_{i=2}^{nN} \powb{\frac{1}{\svai{\optt{\Aca}}{i}}}{2}}{\frac{1}{nN-1}} \\
& \geq \expb{\frac{\Vari{\opt{\sigma}}}{2}-\frac{\Vari{\optt{\sigma}}}{2}}\br{\frac{\sum_{i=2}^{nN}\frac{1}{\powb{\svai{\optt{\Aca}}{i}}{2}}}{nN-1}} \\
& \geq \br{\frac{nN-1}{nN}} \expb{\frac{\Vari{\opt{\sigma}}}{2}-\frac{\Vari{\optt{\sigma}}}{2}} \frac{\FHtwo{N}(\optt{\Kb})}{nN-1}.
\end{align*}
The result follows from simple algebra now.
\end{proof}
\begin{thm}\label{thm:HinfSUB}
Let the solution to the convex optimization and original problem be:
\[\optt{\Kb}=\argmin_{\Kb \in \C} \sum_{i=1}^{nN-1} \svai{\inv{\Acal{\Kb}{N}}}{i} \textrm{ (Convex Opt)} \]
\[\opt{\Kb}=\argmin_{\Kb \in \C} \svama{\Acal{\Kb}{N}} \textrm{ (Original Opt)}\]
respectively. Let $\opt{\Aca}=\inv{\Acal{\opt{\Kb}}{N}}, \optt{\Aca}=\inv{\Acal{\optt{\Kb}}{N}}$. Let
\[\optt{\sigma}=\left[\frac{\svai{\optt{\Aca}}{nN-1}}{\svai{\optt{\Aca}}{1}},\ldots,\frac{\svai{\optt{\Aca}}{1}}{\svai{\optt{\Aca}}{1}}\right]\] \[\opt{\sigma}=\left[\frac{\svai{\opt{\Aca}}{nN-1}}{\svai{\opt{\Aca}}{nN-1}},\ldots,\frac{\svai{\opt{\Aca}}{1}}{\svai{\opt{\Aca}}{nN-1}}\right]\]
Then,
\[\frac{\FHinf{N}(\optt{\Kb})}{\FHinf{N}(\opt{\Kb})} \leq \expb{(nN-1)\br{\frac{\Vari{\opt{\sigma}}-\Vari{\optt{\sigma}}}{2}}}\]
\end{thm}
\begin{proof}
The proof follows a similar structure as the previous theorem and relies on Holder's defect formula. Let $c=\prod_{t=0}^{N-1} \detb{D_t}$. Using the same inequalities with $a_i=\svai{\opt{\Aca}}{i},i=1,2,\ldots,nN-1$, we get
\begin{align*}
& \powb{\FHinf{N}(\opt{\Kb})}{\frac{1}{nN-1}}= c\powb{\prod_{i=1}^{nN-1} \svai{\opt{\Aca}}{i}}{\frac{1}{nN-1}} \\
& \geq \frac{c\expb{-\frac{\Vari{\opt{\sigma}}}{2}}}{nN-1}\left(\sum_{i=1}^{nN-1} {\svai{\opt{\Aca}}{i}}\right)
\end{align*}
where $c=\powb{\prod_{t=0}^{N-1}\detb{D_t}}{\frac{2}{nN-1}}$. Since $\optt{\Kb}$ minimizes $\sum_{i=1}^{nN-1} \svai{\inve{\Acal{\Kb}{N}}}{i}$, we have
\begin{align*}
& \powb{\FHinf{N}(\opt{\Kb})}{\frac{1}{nN-1}} \geq \frac{c\expb{\frac{-\Vari{\opt{\sigma}}}{2}}}{nN-1}\br{\sum_{i=1}^{nN-1}\svai{\optt{\Aca}}{i}} \\
& \geq c\expb{\frac{\Vari{\optt{\sigma}}}{2}-\frac{\Vari{\opt{\sigma}}}{2}}\powb{\prod_{i=1}^{nN-1}\svai{\optt{\Aca}}{i}}{\frac{1}{nN-1}} \\
& = \expb{\frac{\Vari{\optt{\sigma}}}{2}-\frac{\Vari{\opt{\sigma}}}{2}} \powb{\FHinf{N}(\optt{\Kb})}{\frac{1}{nN-1}}.
\end{align*}
The result follows from simple algebra now.
\end{proof}
\subsection{INTERPRETATION OF BOUNDS}
The bounds have the following interpretation: Since the product of singular values is constrained to be fixed, stable systems (with small $\Htwo,\Hinf$ norm) would have all of their singular values close to each other. Thus, if the singular values at the solution discovered by our algorithm are close to each other, we can expect that our solution is close to the true optimum. Further, the bounds say that the only thing that matters is the spread of the singular values relative to the spread of singular values at the optimal solution. A side-effect of the analysis is that it suggests that the spectral norm of $\inv{\Acal{\Kb}{N}}$ be used as a surrogate for the $\FHtwo{N}$ norm and the nuclear norm be a surrogate for the $\FHinf{N}$ norm, since optimizing these surrogates produces solutions with suboptimality bounds on the original objectives.

Finally note that although the bounds depend on the (unknown) optimal solution $\opt{\Kb}$, we can still get a useful bound for the $\FHtwo{N}$ case by simply dropping the effect of the negative term so that
\[\frac{\FHtwo{N}(\optt{\Kb})}{\FHtwo{N}(\opt{\Kb})} \leq \br{\frac{nN}{nN-1}}\expb{\frac{\Vari{\optt{\sigma}}}{2}}.\]
which can be computed after solving the convex problem to get $\optt{\Kb}$. A finer analysis may be possible by looking at the minimum possible value of $\Vari{\opt{\sigma}}$, just based on the block-bidiagonal structure of the matrix $\inve{\Acal{\Kb}{N}}$, but we leave this for future work.

\section{ALGORITHMS AND COMPUTATION} \label{sec:Algos}

In this paper, our primary focus is to discuss the properties of the new convex formulation of structured controller synthesis we developed here. Algorithms for solving the resulting convex optimization problem \eqref{eq:ObjGen} is a topic we will investigate in depth in future work. In most cases, problem \eqref{eq:ObjGen} can be reformulated as a semidefinite programming problem and solved using off-the-shelf interior point methods. However, although theoretically polynomial time, off-the-shelf solvers tend to be inefficient in practice and do not scale. In this section, we lay out some algorithmic options including the one we used in our numerical experiments (section \ref{sec:NumExpt}).

When the objective used is the nuclear norm, $\sum_{i=1}^{nN} \svai{\inv{\Acal{\Kb}{N}}}{i}$, we show that it is possible to optimize the objective using standard Quasi-Newton approaches. The nuclear norm is a nonsmooth function in general, but given the special structure of the matrices appearing in our problem, we show that it is differentiable. For a matrix $X$, the subdifferential of the  nuclear norm $\norm{X}_*$ at $X$ is given by
\[\{U\tran{V}+W:\tran{U}W=0 \text{ or } WV=0,\norm{W}_2\leq 1 \}\]
where $X=U\Sigma \tran{V}$ is the singular value decomposition of $X$. For our problem $X=\inve{\Acal{\Kb}{N}}$, which has a non-zero determinant and hence is a nonsingular square matrix irrespective of the value of $\Kb$. Thus, the subdifferential is a singleton ($\tran{U}W=0 \implies W=0$ as $U$ is full rank and square). This means that the nuclear norm is a differentiable function in our problem and one can use standard gradient descent and Quasi Newton methods to minimize it. These methods are orders of magnitude more efficient than other approaches (reformulating as an SDP and using off-the-shelf interior point methods). They still require computing the SVD of an $nN \times nN$ matrix at every iteration, which will get prohibitively expensive when $nN$ is of the order of several thousands. However, the structure of $\tran{\Acal{\Kb}{N}}\Acal{\Kb}{N}$ is block-tridiagonal and efficient algorithms have been proposed for computing the eigenvalues of such matrices (see \cite{EigenTriSVD} and the references therein). Since the singular values of $\Acal{\Kb}{N}$ are simply square roots of eigenvalues of $\tran{\Acal{\Kb}{N}}\Acal{\Kb}{N}$, this approach could give us efficient algorithms for computing the SVD of $\Acal{\Kb}{N}$.

When the objective is the spectral norm $\svama{\inv{\Acal{\Kb}{N}}}$, we can reformulate the problem as a semidefinite programming problem (SDP):
\begin{align*}
& \mini_{t,\Kb  \in \C} t+\Reg{\Kb} \\
& \text{Subject to } tI \geq \begin{pmatrix} 0 & \tran{\inv{\Acal{\Kb}{N}}} \\ \inv{\Acal{\Kb}{N}} & 0 \end{pmatrix}
\end{align*}
The log-barrier for the semidefinite constraint can be rewritten as $\logb{\detb{t^2-\tran{\inve{\Acal{\Kb}{N}}}\inve{\Acal{\Kb}{N}}}}$ using Schur complements. The matrix $\tran{\inv{\Acal{\Kb}{N}}}\inv{\Acal{\Kb}{N}}$ is a symmetric positive definite block-tridiagonal matrix, which is a special case of a chordal sparsity pattern \cite{andersen2010implementation}. This means that computing the gradient and Newton step for the log-barrier is efficient, with complexity growing as $O(N)$. Thus, at least for the case where the objective is the spectral norm, we can develop efficient interior point methods.
\section{NUMERICAL EXPERIMENTS} \label{sec:NumExpt}
\subsection{Comparing Algorithms: Decentralized Control}
In this section, we compare different approaches to controller synthesis. We work with discrete-time LTI systems over a fixed horizon $N$ with $A_t=A,B_t=B=I,D_t=D=I$. Further, we will use $\C=\{K: K_{ij} =0 \not\in S\}$, where $S$ is the set of non-zero indices of $K$.  The control design methodologies we compare are:\\
\emph{\Mih:} This refers to nonconvex approaches for both the $\FHtwo{N}$ and $\FHinf{N}$ norms. The $\FHtwo{N}$ norm is a differentiable function and we use a standard LBFGS method \cite{schmidt2012minfunc} to minimize it. The $\FHinf{N}$ norm is nondifferentiable, but only at points where the maximum singular value of $\Acal{\Kb}{N}$ is not unique. We use a nonsmooth Quasi Newton method \cite{lewis2012nonsmooth} to minimize it (using the freely available software implementation HANSO \cite{HANSO}).\\
\emph{\CON:} The convex control synthesis described here. In the experiments described here, we use the following objective:
\begin{align}
& \mini_{K}  \frac{1}{nN}\left(\sum_{i=1}^{m}\svai{\inve{\Acal{K}{N}}}{i}\right) \label{eq:Expt} \\
& \text{where} \nonumber \\
& \inve{\Acal{K}{N}} = \left[\begin{matrix}
  I             & 0          & \ldots & \quad 0 \\
  -(A+BK)        & I         & \ldots & \quad 0 \\
  0             & -(A+BK)    & \ldots & \quad 0 \\
  \vdots        & \vdots     & \cdots & \quad \vdots\\
  0             & 0          & \ldots & \quad I
 \end{matrix}\right] \nonumber\\
& \text{Subject to } \nonumber\\
 & K_{ij} = 0 \quad \forall (i,j) \not \in S \nonumber
\end{align}
with $m=nN-1$ as a surrogate for the $\FHinf{N}$ norm and $m=1$ for the $\FHtwo{N}$ norm. Although these objectives are non differentiable, we find that an off-the-shelf LBFGS optimizer \cite{schmidt2012minfunc} works well and use it in our experiments here. \\
\emph{\OPT}: The optimal solution to the problem in the absence of the constraint $\C$. This is simply the solution to a standard LQR problem for the $\FHtwo{N}$ case. For the $\FHinf{N}$ norm, this is computed by solving a series of LQ games with objective:
\[\sum_{t=1}^N \tran{\s_t}{\s_t}-\sum_{t=0}^{N-1} \gamma^2\tran{\noise_t}\noise_t\]
where the controller chooses $\ug$ to minimize the cost while an adversary chooses $\noise_t$ so as to maximize the cost. There is critical value of $\gamma$ below which the upper value of this game is unbounded. This critical value of $\gamma$  is precisely the $\FHinf{N}$ norm and the resulting policies for the controller at this value of $\gamma$ is the $\FHinf{N}$-optimal control policy. For any value of $\gamma$, the solution of the game can be computed by solving a set of Ricatti equations backward in time \cite{bacsar2008h}. \\
We work with a dynamical system formed by coupling a set of systems with unstable dynamics $A^i \in \R^{2 \times 2}$.
\[\s^{i}_{t+1} = A^i\s^{i}_t+\sum_j \eta_{ij}\s^{j}_t+\ug^{i}_t+\noise^{i}_t\]
where $\s^{i}$ denotes the state of the $i$-th system and $\eta_{ij}$ is a coupling coefficient between systems $i$ and $j$. The objective is to design controls $\ug=\{\ug^{i}\}$, in order to stabilize the overall system. In our examples, we use $N=5$ systems giving us a 10 dimensional state space. The $A^{i},\eta_{ij}$ are generated randomly, with each entry having a Gaussian distribution with mean $0$ and variance $10$. The sparsity pattern $S$ is also generated randomly by picking $20\%$ of the off-diagonal entries of $K$ and setting them to $0$. For both the $\CON,\Mih$ problems, we initialize the optimizer at the same point $\Kb=0$. 
\begin{figure}[htb]
\begin{center}
\includegraphics[width=.65\columnwidth]{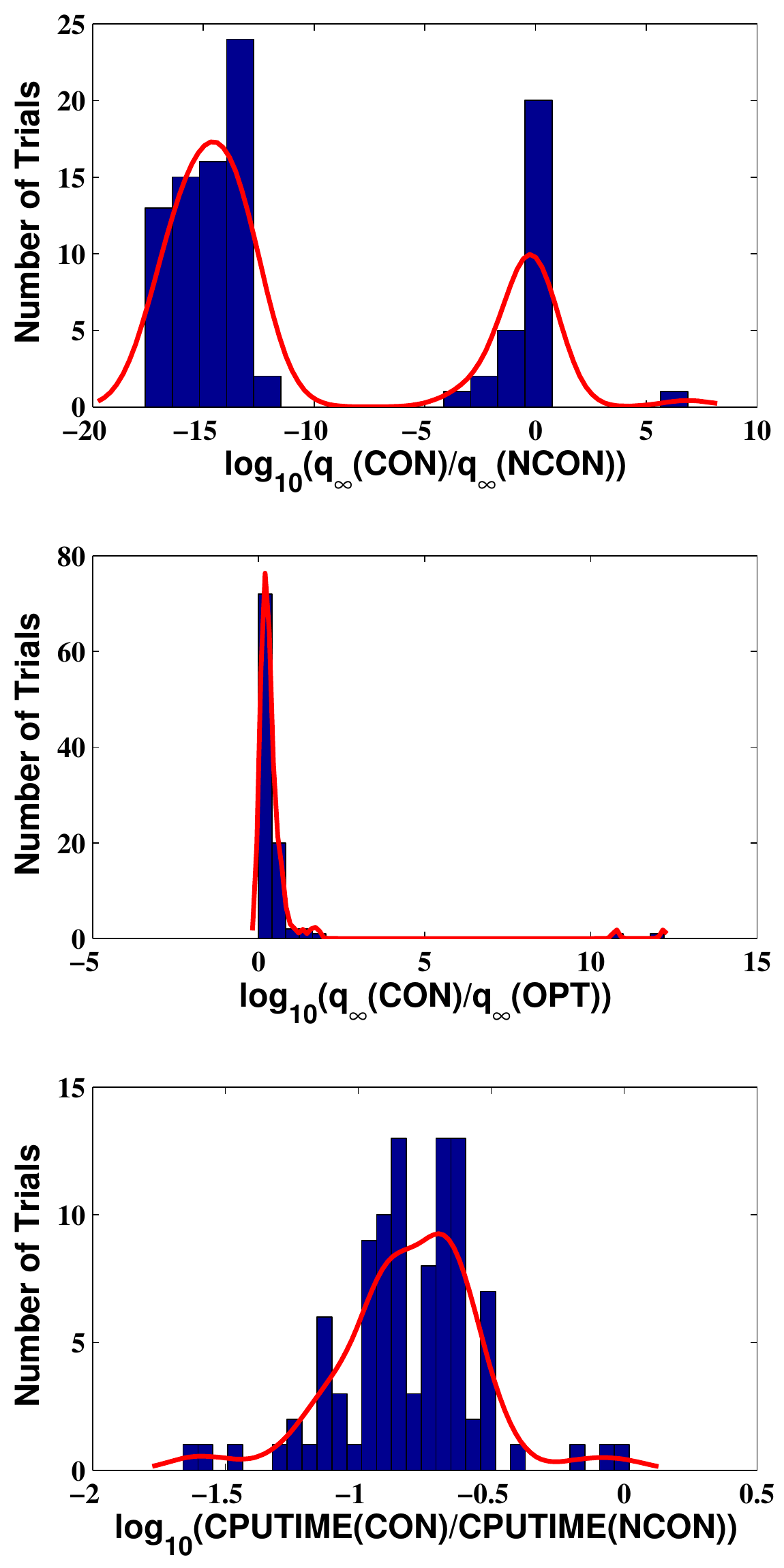}
\caption{Comparison of Algorithms for $\FHinf{N}$-norm Controller Synthesis. The blue bars represent histograms and the red curves kernel density estimates of the distribution of values.}\label{fig:ComparePoints}
\end{center}
\end{figure}
\begin{figure}[htb]
\begin{center}
\includegraphics[width=.65\columnwidth]{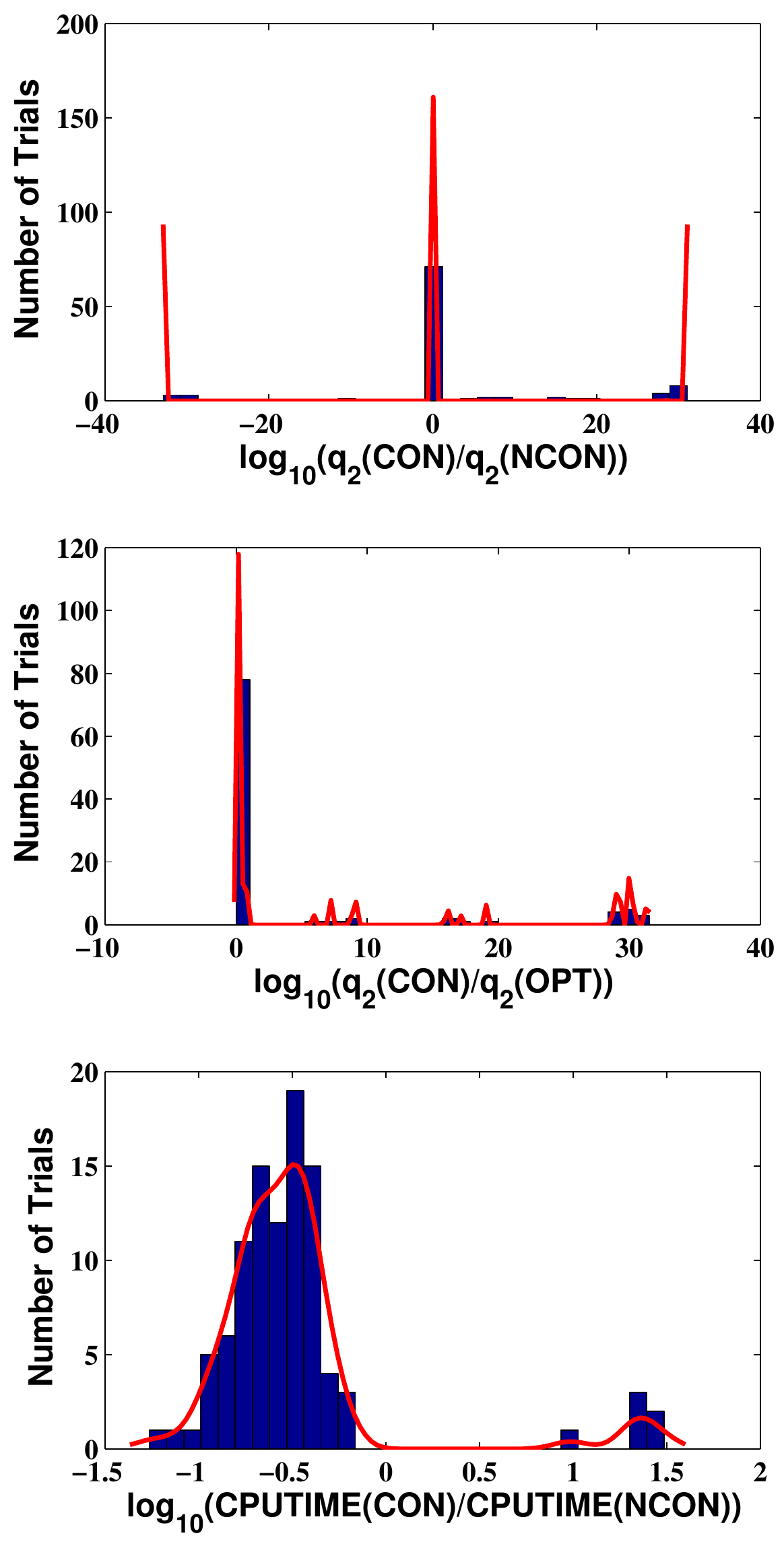}
\caption{Comparison of Algorithms for $\FHtwo{N}$-norm Controller Synthesis. The blue bars represent histograms and the red curves kernel density estimates of the distribution of values.}\label{fig:ComparePoints2}
\end{center}
\end{figure}
For the $\FHinf{N}$ norm, we present results comparing the approaches over 100 trials. The $\FHinf{N}$ norm of the solution obtained by the $\CON$  approach to that found by $\Mih,\OPT$ in figure \ref{fig:ComparePoints}. We plot histograms of how the $\FHinf{N}$ compares between the $\CON,\Mih$ and $\OPT$ approaches. The red curves show kernel-density estimates of the distribution of values being plotted. The results show that $\CON$ consistently outperforms $\Mih$ and often achieves performance close to the centralized $\OPT$ solution. The x-axis denotes the ratio between objectives on a log scale. The y-axis shows the frequency with which a particular ratio is attained (out of a 100 trials). We also plot a histogram of computation times with the log of ratio of CPU times for the \CON~and \Mih~algorithms on the x-axis. Again, in terms of CPU times, the \CON~approach is consistently superior except for a small number of outliers. For the $\FHtwo{N}$ norm, we plot the results in figure \ref{fig:ComparePoints2}. Here, the \Mih~approach does better and beats the \CON~approach for most trials. However, in more than $70\%$ of the trials the $\FHtwo{N}$ norm of the solution found by \CON~is within $2\%$ of that found by \Mih. In terms of computation time, the \CON~approach retains superiority.

The numerical results indicate that the convex surrogates work well in many cases. However, they do fail in particular cases. In general, the surrogates seem to perform better on the $\FHinf{N}$ norm than the $\FHtwo{N}$ norm. The initial results are promising but we believe that further analytical and numerical work is required to exactly understand when the convex objectives proposed in this paper are good surrogates for the original nonconvex $\FHtwo{N}$ and $\FHinf{N}$ objectives.

\section{GENERALIZATION TO NONLINEAR SYSTEMS}
We now present a generalization of our approach to nonlinear systems. The essential idea is to study a nonlinear system in terms of sensitivities of system trajectories with respect to disturbances. Consider a control-affine nonlinear discrete-time system:
\begin{align*}
& \s_1=D_0\noise_0 \\
& \s_{t+1}=\Dyna_t(\s_t)+\DynB_t(\s_t)\ug_{t}+D_{t}\noise_{t} \quad (1\leq t \leq N-1)
\end{align*}
where $\Dyna_t : \R^n \mapsto \R^n$ and $\DynB : \R^n \mapsto \R^{n \times \nc}$, $D_{t} \in \R^{n \times n}, \s_t \in \R^n, w_t \in \R^n, \ug_t \in \R^{\nc}$. Suppose that $0$ is an equilibrium point (if not, we simply translate the coordinates to make this the case). Now we seek to design a controller $u_t=K_t\phi(\s_t)$ where $\phi$ is any set of fixed ``features'' of the state on which we want the control to depend that minimizes deviations from the constant trajectory $[0,0,\ldots,0]$. We can look at the closed loop system:
\[\s_{t+1}=\Dyna_t(\s_t)+\DynB_t(\s_t)K_t\Dynf_t(\s_t)+D_{t}\noise_{t}\]
where $\Dynf_t(\s_t) \in \R^{\nm}, K_t \in \R^{\nc \times \nm}$. As before, let $\Kb=\{K_t: 1\leq t \leq N-1\}.$ Let $\Acal{\Kb}{N}(\wtraj)$ denote the (nonlinear) mapping from a sequence of disturbances $\wtraj=[\noise_0,\ldots,\noise_{N-1}]$ to the state space trajectory $\traj=[\s_1,\ldots_,\s_N]$. The finite-horizon $\FHinf{N}$ norm for a nonlinear system can be defined analogously as for a linear system.
\begin{align}
\max_{\wtraj \neq 0} \frac{\norm{\Acal{\Kb}{N}(\wtraj)}}{\norm{\wtraj}}. \label{eq:Hinfnl}
\end{align}
Given a state trajectory $\traj=[\s_1,\ldots,\s_N]$, we can recover the noise sequence as
\begin{align}
& \noise_0=\inve{D_0}\s_1 \nonumber \\
& \noise_t=\inve{D_t}\left(\s_{t+1}-\Dyna_t(\s_{t})-\DynB_t(\s_{t})K_t\Dynf_t(\s_{t})\right),t>0 \label{eq:Nonlin}
\end{align}
Thus the map $\Acal{\Kb}{N}$ is invertible. Let $\inve{\Acal{\Kb}{N}}$ denote the inverse. It can be shown (theorem \ref{thm:NonLinBound}) that the objective \eqref{eq:Hinfnl} (assuming it is finite) can be bounded above by
\[\sup_{\traj}{\powb{\frac{\sum_{i=1}^{nN-1}  \svai{\br{\frac{\partial \inv{\Acal{\Kb}{N}}(\traj)}{\partial \traj}}}{i}}{nN-1}}{nN-1}}\]
In the linear case, the maximization over $\traj$ is unnecessary since the term being maximized is independent of $\traj$. However, for a nonlinear system, the Jacobian of $\inv{\Acal{\Kb}{N}}(\traj)$ is a function of $\traj$ and an explicit maximization needs to be performed to compute the objective. Thus, we can formulate the control design problem as
\begin{align}
\min_{\Kb \in \C} \sup_{\traj}{\powb{\frac{\sum_{i=1}^{nN-1}  \svai{\br{\frac{\partial \inv{\Acal{\Kb}{N}}(\traj)}{\partial \traj}}}{i}}{nN-1}}{nN-1}} \label{eq:NonlinOpt}
\end{align}
The convexity of the above objective follows using a very similar proof as the linear case (see theorem \ref{thm:NonLinConvex}). Computing the objective (maximizing over $\traj$) in general would be a hard problem, so this result is only of theoretical interest in its current form. However, in future work, we hope to explore the computational aspects of this formulation more carefully. 

\section{Discussion and Related Work}

There have been three major classes of prior work in synthesizing structured controllers: Frequency domain approaches, dynamic programming and nonconvex optimization approaches. We compare the relative merits of the different approaches in this section.

In frequency domain approaches, problems are typically formulated as follows:
\begin{align*}
\mini_{K} & \parallel\textrm{Closed loop system with feedback K}\parallel \\
\text{Subject to } & K \text{ Stabilizing }, K \in \C
\end{align*}
where $\norm{\cdot}$ is typically the $\Htwo$ or $\Hinf$ norm. In general, these are solved by reparameterizing the problem in terms of a Youla parameter (via a nonlinear transformation), and imposing special conditions on $\C$ (like quadratic invariance) that guarantee that the constraints $\C$ can be translated into convex constraints on the Youla parameter \cite{rotkowitz2006}\cite{qi2004structured}. There are multiple limitations of these approaches:\\
(1) Only specific kinds of constraints can be imposed on the controller. Many of the examples have the restriction that the structure of the controller mirrors that of the plant.\\
(2) They result in infinite dimensional convex programs in general. One can solve them using a sequence of convex programming problems, but these approaches are susceptible to numerical issues and the degree of the resulting controllers may be ill-behaved, leading to practical problems in terms of implementing them.\\
(3) The approaches rely on frequency domain notions and cannot handle time-varying systems.\\
In the special case of poset-causal systems (where the structure of the plant and controller can be described in terms of a partial order \cite{shah2013}), the problem can be decomposed when the performance metric is the $\Htwo$ norm and explicit state-space solutions are available by solving Ricatti equations for subsystems and combining the results. For the $\Hinf$ norm, a state-space solution using an LMI approach was developed in \cite{Scherer}.

Another thread of work on decentralized control looks at special cases where dynamic programming techniques can be used in spite of the decentralization constraints. The advantage of these approaches is that they directly handle finite horizon and time-varying approaches. For the LEQG cost-criterion, a dynamic programming approach was developed in \cite{fan1994centralized} for the case of 1-step delay in a 2-agent decentralized control problem. In \cite{swigart2010explicit}, the authors show that for the case of 2 agents (a block-lower triangular structure in $A,B$ with 2 blocks) can be solved via dynamic programming. In \cite{lamperski2013optimal}, the authors develop a dynamic programming solution that generalizes this and applies to general ``partially-nested'' systems allowing for both sparsity and delays.

All the above methods work for special structures on the plant and controller (quadratic invariance/partial nestedness) under which decentralized controllers can be synthesized using either convex optimization or dynamic programming methods.

In very recent work \cite{Lavaei2013}, the authors pose decentralized control (in the discrete-time, finite horizon, linear quadratic setting) as a rank-constrained semidefinite programming problem. By dropping the rank constraint, one can obtain a convex relaxation of the problem. The relaxed problem provides a solution to the original problem only when the relaxed problem has a rank-1 solution. However, it is unknown when this can be guaranteed, and how a useful controller can be recovered from a higher-rank solution. Further, the SDP posed in this work grows very quickly with the problem dimension.

Our work differs from these previous works in one fundamental way: Rather than looking for special decentralization structures that can be solved tractably under standard control objectives, we formulate a new control objective that helps us solve problems with \emph{arbitrary} decentralization constraints. In fact, we can handle \emph{arbitrary convex constraints} - decentralization constraints that impose a sparsity pattern on $\Kb$ are a special case of this. We can also handle time-varying linear systems. Although the objective is nonstandard, we have provided theoretical and numerical evidence that it is a sensible control objective. The only other approaches that handle all these problems are nonconvex approaches \cite{zhai2001decentralized,apkarian2008mixed,linfarjovTAC13admm}. We have shown that our approach outperforms a standard nonconvex approach, both in terms of performance of resulting controller and in computation times.

We also believe that this was the first approach to exploit a fundamental limitation (Bode's sensitivity integral) to develop efficient control design algorithms. The fact that the spectrum of the input output map satisfies a conservation law (the sum of the logs of singular values is fixed) is a limitation which says that reducing some of the singular values is bound to increase the others. However, this limitation allows us to approximate the difficult problem of minimizing the $\Htwo$ or $\Hinf$ norm with the easier problem of minimizing a convex surrogate, leading to efficient solution.

\section{CONCLUSION}

We have argued that the framework developed seems promising and overcomes limitations of previous works on computationally tractable approaches to structured controller synthesis. Although the control objective used is non-standard, we have argued why it is a sensible objective, and we also presented numerical examples showing that it produces controllers outperforming other nonconvex approaches. Further, we proved suboptimality bounds that give guidance on when our solution is good even with respect to the original ($\Htwo/\Hinf$) metrics. There are three major directions for future work: 1) Investigating the effect of various objectives in our family of control objectives, 2) Developing efficient solvers for the resulting convex optimization problems and 3) Deriving computationally efficient algorithms for nonlinear systems.


\section*{Acknowledgements}
This work was supported by the NSF. We would like to thank Mehran Mesbahi, Andy Lamperski and Parikshit Shah for helpful discussions and feedback on earlier versions of this manuscript. We would also like to thank Mihailo Jovanovic for suggestions that significantly improved the numerical results section.

\bibliographystyle{IEEEtran}
\bibliography{Ref}

\begin{thebibliography}{10}
\providecommand{\url}[1]{#1}
\csname url@samestyle\endcsname
\providecommand{\newblock}{\relax}
\providecommand{\bibinfo}[2]{#2}
\providecommand{\BIBentrySTDinterwordspacing}{\spaceskip=0pt\relax}
\providecommand{\BIBentryALTinterwordstretchfactor}{4}
\providecommand{\BIBentryALTinterwordspacing}{\spaceskip=\fontdimen2\font plus
\BIBentryALTinterwordstretchfactor\fontdimen3\font minus
  \fontdimen4\font\relax}
\providecommand{\BIBforeignlanguage}[2]{{%
\expandafter\ifx\csname l@#1\endcsname\relax
\typeout{** WARNING: IEEEtran.bst: No hyphenation pattern has been}%
\typeout{** loaded for the language `#1'. Using the pattern for}%
\typeout{** the default language instead.}%
\else
\language=\csname l@#1\endcsname
\fi
#2}}
\providecommand{\BIBdecl}{\relax}
\BIBdecl

\bibitem{kalman1960new}
R.~E. Kalman \emph{et~al.}, ``A new approach to linear filtering and prediction
  problems,'' \emph{Journal of basic Engineering}, vol.~82, no.~1, pp. 35--45,
  1960.

\bibitem{witsenhausen1968counterexample}
H.~S. Witsenhausen, ``A counterexample in stochastic optimum control,''
  \emph{SIAM Journal on Control}, vol.~6, no.~1, pp. 131--147, 1968.

\bibitem{blondel1997np}
V.~Blondel and J.~N. Tsitsiklis, ``Np-hardness of some linear control design
  problems,'' \emph{SIAM Journal on Control and Optimization}, vol.~35, no.~6,
  pp. 2118--2127, 1997.

\bibitem{rotkowitz2002decentralized}
M.~Rotkowitz and S.~Lall, ``Decentralized control information structures
  preserved under feedback,'' in \emph{Proceedings of the 41st IEEE Conference
  on Decision and Control}, vol.~1, 2002, pp. 569--575.

\bibitem{qi2004structured}
X.~Qi, M.~V. Salapaka, P.~G. Voulgaris, and M.~Khammash, ``Structured optimal
  and robust control with multiple criteria: A convex solution,'' \emph{IEEE
  Transactions on Automatic Control}, vol.~49, no.~10, pp. 1623--1640, 2004.

\bibitem{rotkowitz2006}
.~Rotkowitz, Michael and S.~Lall, ``A characterization of convex problems in
  decentralized control,'' \emph{IEEE Transactions on Automatic Control},
  vol.~51, no.~2, pp. 274--286, 2006.

\bibitem{shah2013}
P.~Shah, ``H2-optimal decentralized control over posets: A state-space solution
  for state-feedback,'' \emph{IEEE Transactions on Automatic Control}, vol.~PP,
  no.~99, pp. 1--1, 2013.

\bibitem{fan1994centralized}
C.-H. Fan, J.~L. Speyer, and C.~R. Jaensch, ``Centralized and decentralized
  solutions of the linear-exponential-gaussian problem,'' \emph{IEEE
  Transactions on Automatic Control}, vol.~39, no.~10, pp. 1986--2003, 1994.

\bibitem{swigart2010explicit}
J.~Swigart and S.~Lall, ``An explicit dynamic programming solution for a
  decentralized two-player optimal linear-quadratic regulator,'' in
  \emph{Proceedings of mathematical theory of networks and systems}, 2010.

\bibitem{lamperski2013optimal}
A.~Lamperski and L.~Lessard, ``Optimal decentralized state-feedback control
  with sparsity and delays,'' \emph{arXiv preprint arXiv:1306.0036}, 2013.

\bibitem{linfarjovTAC13admm}
F.~Lin, M.~Fardad, and M.~R. Jovanovi\'c, ``Design of optimal sparse feedback
  gains via the alternating direction method of multipliers,'' \emph{IEEE
  Trans. Automat. Control}, vol.~58, no.~9, pp. 2426--2431, September 2013.

\bibitem{apkarian2008mixed}
P.~Apkarian, D.~Noll, and A.~Rondepierre, ``Mixed h-2/h-inf control via
  nonsmooth optimization,'' \emph{SIAM Journal on Control and Optimization},
  vol.~47, no.~3, pp. 1516--1546, 2008.

\bibitem{burke2006hifoo}
J.~Burke, D.~Henrion, A.~Lewis, and M.~Overton, ``Hifoo-a matlab package for
  fixed-order controller design and h∞ optimization,'' in \emph{Fifth IFAC
  Symposium on Robust Control Design, Toulouse}, 2006.

\bibitem{iglesias2001tradeoffs}
P.~A. Iglesias, ``Tradeoffs in linear time-varying systems: an analogue of
  bode's sensitivity integral,'' \emph{Automatica}, vol.~37, no.~10, pp.
  1541--1550, 2001.

\bibitem{Dj2013}
K.~Dvijotham, E.~Todorov, and M.~Fazel, ``Convex controller design via
  covariance minimization,'' in \emph{Proceedings of the Allerton Conference on
  Control and Computing}, 2013.

\bibitem{blondel2000survey}
V.~D. Blondel and J.~N. Tsitsiklis, ``A survey of computational complexity
  results in systems and control,'' \emph{Automatica}, vol.~36, no.~9, pp.
  1249--1274, 2000.

\bibitem{lewis1995convex}
A.~S. Lewis, ``The convex analysis of unitarily invariant matrix functions,''
  \emph{Journal of Convex Analysis}, vol.~2, no.~1, pp. 173--183, 1995.

\bibitem{becker2012variance}
R.~Becker, ``The variance drain and jensen's inequality,'' \emph{CAEPR Working
  Paper}, 2012.

\bibitem{EigenTriSVD}
A.~{Sandryhaila} and J.~M.~F. {Moura}, ``{Eigendecomposition of Block
  Tridiagonal Matrices},'' \emph{ArXiv e-prints Arxiv:1306.0217}, Jun. 2013.

\bibitem{andersen2010implementation}
M.~S. Andersen, J.~Dahl, and L.~Vandenberghe, ``Implementation of nonsymmetric
  interior-point methods for linear optimization over sparse matrix cones,''
  \emph{Mathematical Programming Computation}, vol.~2, no.~3, pp. 167--201,
  2010.

\bibitem{schmidt2012minfunc}
M.~Schmidt, ``minfunc: unconstrained differentiable multivariate optimization
  in matlab,'' 2012.

\bibitem{lewis2012nonsmooth}
A.~S. Lewis and M.~L. Overton, ``Nonsmooth optimization via quasi-newton
  methods,'' \emph{Mathematical Programming}, pp. 1--29, 2012.

\bibitem{HANSO}
M.~Overton, ``Hanso {http://www.cs.nyu.edu/overton/software/hanso/}, 2013.''

\bibitem{bacsar2008h}
T.~Ba{\c{s}}ar and P.~Bernhard, \emph{H-infinity optimal control and related
  minimax design problems: a dynamic game approach}.\hskip 1em plus 0.5em minus
  0.4em\relax Springer, 2008.

\bibitem{Scherer}
C.~W. {Scherer}, ``{Structured Hinf-Optimal Control for Nested
  Interconnections: A State-Space Solution},'' \emph{ArXiv e-prints:
  arXiv:1305.1746}, May 2013.

\bibitem{Lavaei2013}
J.~Lavaei, ``Optimal decentralized control problem as a rank-constrained
  optimization,'' in \emph{Proceedings of the Allerton Conference on Control
  and Computing}, 2013.

\bibitem{zhai2001decentralized}
G.~Zhai, M.~Ikeda, and Y.~Fujisaki, ``Decentralized h-2/h-inf controller
  design: a matrix inequality approach using a homotopy method,''
  \emph{Automatica}, vol.~37, no.~4, pp. 565--572, 2001.

\end{thebibliography}

\section{APPENDIX} \label{sec:Appendix}

\subsection{Penalizing Control Effort} \label{sec:PenControl}


A more direct approach is to augment the state to include the controls. We define an augmented problem with $\bar{\s_t} \in \R^{\ns+\nc},\overline{\noise}_t \in \R^{\ns+\nc}$.
\begin{align*}
& \overline{A}_t = \begin{pmatrix} A_t & 0 \\ 0 & 0 \end{pmatrix}, \overline{B}_t = \begin{pmatrix} B_t \\ R_t \end{pmatrix}, \overline{D}_t = \begin{pmatrix} D_t & 0 \\ 0 & \gamma I \end{pmatrix} \\
& \overline{\s}_{t+1} = \overline{A}_t\overline{\s}_{t}+\overline{B}_t\ug_t+\overline{D}_t\overline{\noise}_t
\end{align*}
Partitioning the new state $\overline{\s}_t=\begin{pmatrix} \s_t \\ \tilde{\s}_t \end{pmatrix}$, $\overline{\noise}_t=\begin{pmatrix} \noise_t \\ \tilde{\noise}_t \end{pmatrix}$, we have:
\begin{align*}
\s_{t+1}=A_t\s_t+B_t\ug_t+D_t\noise_t, \tilde{\s}_{t+1} = R_t\ug_t+\gamma \tilde{\noise}_t
\end{align*}
Given this,
\begin{align*}
& \sum_{t=1}^N \tran{\overline{\s}_t}\overline{\s}_t = \sum_{t=1}^N \tran{\s_t}\s_t+\\
& \quad \sum_{t=1}^{N-1}\tranb{R_t\ug_t+\gamma\tilde{\noise}_{t}}\br{R_t\ug_t+\gamma\tilde{\noise}_{t}}+\gamma^2 \tran{\noise_0}\noise_0
\end{align*}
 In the limit $\gamma \to 0$, we recover the standard LQR cost. However, setting $\gamma = 0$ violates the condition of invertibility. Thus, solving the problem with an augmented state $\overline{\s} \in \R^{\nc+\nu}$, $\overline{\noise} \in \R^{\nc+\nu}$, \[\overline{A}_t=\begin{pmatrix} A_t & 0 \\ 0 & 0 \end{pmatrix}, \overline{B}_t=\begin{pmatrix} B_t \\ R_t \end{pmatrix}, \overline{D}_t=\begin{pmatrix} D_t & 0\\0 &  \gamma I \end{pmatrix}\]
solves the problem with a quadratic control cost in the limit $\gamma \to 0$. The caveat is that the problems \eqref{eq:ObjHtwo}\eqref{eq:ObjHinf} become increasingly ill-conditioned as $\gamma \to 0$. However, we should be able to solve the problem for a small value of $\gamma$, which models the quadratic controls cost closely but still leads to a sufficiently well-conditioned problem that we can solve numerically.

\subsection{Dynamic Output Feedback} \label{sec:DynOutput}

So far, we have described the problem in terms of direct state feedback $\ug_t=K_t \s_t$. However, we can also model output feedback $u_t=K_tC_t\s_t$ by simply defining $\tilde{K_t}=K_tC_t$ where $C_t \in \R^{\nm \times \ns}$ is a measurement matrix that produces $\nm$ measurements given the state. Convex constraints on $K_t$ will translate into convex constraints on $\tilde{K_t}$, since $\tilde{K_t}$ is a linear function of $K_t$. If we wanted to allow our controls to depend on the previous $\nh$ measurements (dynamic output feedback), we simply create an augmented state $\overline{\s_t}=\begin{pmatrix}\s_t \\ \vdots \\ \s_{t-\nh} \end{pmatrix}$. Then, we can define $K_t \in \R^{\nc \times \nh\nm}$ and \[\tilde{K_t}=K_t \begin{pmatrix}C_t & 0 & \ldots & 0 \\ 0 & C_{t-1} & \ldots & 0 \\ \vdots & \vdots & \vdots & \vdots \\ 0 & 0 & \ldots & C_{t-\nh} \end{pmatrix}\]
and an augmented dynamics

\begin{align*}
& \overline{A}_t = \begin{pmatrix}A_t & 0 & \ldots & 0 & 0 \\ I & 0 & \ldots & 0 & 0 \\ \vdots & \vdots & \vdots & \vdots & \vdots \\0 & 0 & \ldots & I & 0 \end{pmatrix}, \overline{B}_t = \begin{pmatrix} B_t \\ 0 \\ \vdots \\ 0\end{pmatrix} \\
 & \overline{D}_t = \begin{pmatrix}D_t & 0 & \ldots & 0 & 0 \\ 0 & \gamma I & \ldots & 0 & 0 \\ \vdots & \vdots & \vdots & \vdots & \vdots \\0 & 0 & \ldots & 0 & \gamma I \end{pmatrix}
\end{align*}

Again, we need to set $\gamma=0$ to exactly match the standard output feedback problem but that violates the assumption of invertibility. We can consider taking $\gamma \to 0$ and recovering the solution as a limiting case, as in the previous section. 


\subsection{Proofs}

\begin{thm}\label{thm:HinfUB}
\begin{align*}
&\FHinf{N}(\Kb)=\svama{\Acal{\Kb}{N}} \\
& \leq \prod_{t=0}^{N-1} \detb{D_t}\powb{\frac{\sum_{i=1}^{\ns N-1} \svai{\inve{\Acal{\Kb}{N}}}{i}}{\ns N-1}}{\ns N-1}
\end{align*}
\end{thm}
\begin{proof}
Since $\Acal{\Kb}{N}$ is a block lower triangular matrix (a reflection of the fact that we have a causal linear system) , its determinant is simply the product of determinants of diagonal blocks:  $\detb{\Acal{\Kb}{N}}=\prod_t \detb{D_t}=c$ \emph{independent} of the values of $\Ad_t$. In fact, this result is a generalization of Bode's classical sensitivity integral result and has been studied in \cite{iglesias2001tradeoffs}. Since the product of singular values is equal to the determinant, we have
\begin{align*}
&\svama{\Acal{\Kb}{N}} = \frac{c}{\displaystyle\prod_{i=2}^{nN} \svai{\Acal{\Kb}{N}}{i}} =c\displaystyle\prod_{i=1}^{nN-1} \svai{\inve{\Acal{\Kb}{N}}}{i}
\end{align*}
where the last equality follows because the singular values of $\inve{\Acal{\Kb}{N}}$ are simply reciprocals of the singular values of $\Acal{\Kb}{N}$. The result now follows using the AM-GM inequality.
\end{proof}

\begin{thm}\label{thm:HtwoUB}
\[\FHtwo{N}(\Kb) \leq nN\br{\prod_{t=0}^{N-1} \detb{D_t}}^2\powb{\svama{\inve{\Acal{\Kb}{N}}}}{2(\ns N-1)}\]
\end{thm}
\begin{proof}
Let $\prod_{t=0}^{N-1} \detb{D_t}=c$. From the above argument, we can express $\svai{\Acal{\Kb}{N}}{i}$ as
\begin{align*}
c \prod_{j \neq nN-i+1} \svai{\inv{\Acal{\Kb}{N}}}{j} \leq  c\powb{\svama{\inve{\Acal{\Kb}{N}}}}{(\ns N-1)}.
\end{align*}
The expression for $\FHtwo{N}(\Kb)$ is
\begin{align*}
& \sum_{i=1}^{nN} \powb{\svai{\Acal{\Kb}{N}}{i}}{2} \leq nN c^2\powb{\svama{\inve{\Acal{\Kb}{N}}}}{2(\ns N-1)}.
\end{align*}
\end{proof}

\begin{thm}\label{thm:NonLinConvex}
For the nonlinear system described in \eqref{eq:Nonlin}, the function
\[\sup_{\traj} \powb{\frac{\sum_{i=1}^{nN-1}\svai{\br{\frac{\partial \inv{\Acal{\Kb}{N}}(\traj)}{\partial \traj}}}{i}}{nN-1}}{nN-1}\]
is convex in $\Kb$.
\end{thm}
\begin{proof}
First fix $\wtraj$ to an arbitrary value. From \eqref{eq:Nonlin}, we know that $\br{\frac{\partial \inv{\Acal{\Kb}{N}}(\traj)}{\partial \traj}}$ is of the form
\[\left[\begin{matrix}
  \inve{D_0}    & 0         & \ldots    & \ldots & \quad 0 \\
  -\inve{D_1}\frac{\partial \noise_1}{\partial \s_1}        & \inve{D_1}         & \ldots    & \ldots & \quad 0 \\
  0             & -\inve{D_2}\frac{\partial \noise_2}{\partial \s_2}    & \inve{D_2}         & \ldots & \quad 0 \\
  \vdots        & \vdots    & \vdots    & \cdots & \quad \vdots\\
  0             & 0         & 0         & \ldots & \quad \inve{D_{N-1}}
 \end{matrix}\right]\]
Since $\noise_t=\inve{D_t}\left(\s_{t+1}-\Dyna_t(\s_{t})-\DynB_t(\s_{t})K_t\Dynf_t(\s_{t})\right)$, $\noise_t$ is an affine function of $\Kb$. Hence, so is $\frac{\partial \noise_{t}}{\partial \s_t}$, for any $t$. Thus, the overall matrix $\br{\frac{\partial \inv{\Acal{\Kb}{N}}(\traj)}{\partial \traj}}=M(\Kb)$ is an affine function of $\Kb$. Thus, by composition properties, $\powb{\frac{\sum_{i=1}^{nN-1}\svai{M(\Kb)}{i}}{nN-1}}{nN-1}$ is a convex function of $\Kb$ for any fixed $\traj$. Taking a supremum over all $\traj$ preserves convexity, since the pointwise supremum of a set of convex functions is convex.
\end{proof}

\begin{thm}\label{thm:NonLinBound}
Consider the nonlinear system described in \eqref{eq:Nonlin}. Suppose that $\sup_{\wtraj \neq 0} \frac{\norm{\Acal{\Kb}{N}(\wtraj)}}{\norm{\wtraj}}$ is finite and the supremum is achieved at $\opt{\wtraj} \neq 0$ for all values of $\Kb$. Then, $\sup_{\wtraj \neq 0} \frac{\norm{\Acal{\Kb}{N}(\wtraj)}}{\norm{\wtraj}}$ is bounded above by
\[\sup_{\traj} \powb{\frac{\sum_{i=1}^{nN-1}\svai{\br{\frac{\partial \inv{\Acal{\Kb}{N}}(\traj)}{\partial \traj}}}{i}}{nN-1}}{nN-1}\]
\end{thm}
\begin{proof}
By theorem \ref{thm:Nonlin}, $\sup_{\wtraj \neq 0} \frac{\norm{\Acal{\Kb}{N}(\wtraj)}}{\norm{\wtraj}}$ is bounded above by
\[\sup_{\wtraj \neq 0} \svama{\frac{\partial \Acal{\Kb}{N}(\wtraj)}{\partial \wtraj}}.\]
Now, $M(\Kb)=\frac{\partial \Acal{\Kb}{N}(\wtraj)}{\partial \wtraj}$ is a lower-triangular matrix (since we have a causal system) and the diagonal blocks are given by $D_t$. Thus, $\detb{M(\Kb)}=\prod_{t=0}^{N-1} \detb{D_t}=c$, and we can rewrite  $\svama{M(\Kb)}$ as
$c{\prod_{i=1}^{nN-1}\svai{\inv{M(\Kb)}}{i}}$.
By the rules of calculus, we know that
\[\inv{M(\Kb)} = \br{\frac{\partial \inv{\Acal{\Kb}{N}}(\traj)}{\partial \traj}}_{\traj=\Acal{\Kb}{N}(\wtraj)}\]
Thus, the above objective reduces to
\[\sup_{\wtraj \neq 0} {\prod_{i=1}^{nN-1}\svai{\br{\frac{\partial \inv{\Acal{\Kb}{N}}(\traj)}{\partial \traj}}_{\traj=\Acal{\Kb}{N}(\wtraj)}}{i}}.\]
Given any $\traj$, we can find $\wtraj$ such that $\traj=\Acal{\Kb}{N}(\wtraj)$ (simply choose $\wtraj=\inv{\Acal{\Kb}{N}}(\traj)$). Thus, the above quantity is equal to
\[\sup_{\traj} {\prod_{i=1}^{nN-1}\svai{\br{\frac{\partial \inv{\Acal{\Kb}{N}}(\traj)}{\partial \traj}}}{i}}.\]
The result now follows using the AM-GM inequality.
\end{proof}
\begin{thm}\label{thm:Nonlin}
Let $g(\y):\R^l \mapsto \R^p$ be any differentiable function. If the function $\frac{\norm{g(\y)}_2}{\norm{\y}_2}$ attains its maximum at $\y^\ast$,
\[\sup_{y} \frac{\norm{g(\y)}_2}{\norm{\y}_2} \leq \sup_{\y \neq 0} \svama{\frac{\partial g(\y)}{\partial \y}} \quad \forall \y \neq 0.\]
\end{thm}
\begin{proof}
 $\logb{\frac{\norm{g(\y)}^2}{\norm{\y}^2}}$ is differentiable at any $\y \neq 0$ and hence at $\y=\y^\ast$. Since this is an unconstrained optimization problem, we can write the optimality condition ($0$ gradient):
\[\frac{2\br{\frac{\partial g(\y)}{\partial \y}}_{\y=\y^\ast}g(\y^\ast)}{\norm{g(\y^\ast)}^2}=\frac{2\y^\ast}{\norm{\y^\ast}^2}\]
Taking the $\ell_2$ norm on both sides, we get
\[\frac{\norm{\br{\frac{\partial g(\y)}{\partial \y}}_{\y=\y^\ast}g(\y^\ast)}}{\norm{g(\y^\ast)}}=\frac{\norm{g(\y^\ast)}}{\norm{\y^\ast}}\]
Since $\y^{\ast}$ maximizes $\frac{\norm{g(\y)}}{\norm{\y}}$, for any $\y \neq 0$, we have:
\begin{align*}
& \frac{\norm{g(\y)}}{\norm{\y}} \leq \frac{\norm{g(\y^\ast)}}{\norm{\y^\ast}}=\frac{\norm{\br{\frac{\partial g(\y)}{\partial \y}}_{\y=\y^\ast}g(\y^\ast)}}{\norm{g(\y^\ast)}} \\
& \leq \svama{\frac{\partial g(\y)}{\partial \y}}_{\y=\y^\ast} \leq \max_{\y \neq 0} \svama{\frac{\partial g(\y)}{\partial \y}}.
\end{align*}
Taking supremum over $y$ on both sides, we get the result.
\end{proof}

\end{document}